\documentclass[10pt, conference]{IEEEtran}

\usepackage{braket}
\usepackage{amsmath,amssymb,amsfonts,dsfont,amsthm}
\usepackage{algorithmic}
\usepackage{graphicx}
\usepackage{textcomp}
\usepackage{xcolor}
\usepackage{comment}
\usepackage{cite}
\usepackage{authblk}

\newcommand\blfootnote[1]{%
  \begingroup
  \renewcommand\thefootnote{}\footnote{#1}%
  \addtocounter{footnote}{-1}%
  \endgroup
}

\newtheorem{theorem}{Theorem}
\newtheorem{corollary}[theorem]{Corollary}

\newtheorem{lemma}[theorem]{Lemma}
\newtheorem{definition}[theorem]{Definition}
\newtheorem{proposition}[theorem]{Proposition}

\newlength{\blank}
\newcommand{\EE}{\mathbb{E}}
\newcommand{\RR}{\mathbb{R}}
\newcommand{\NN}{\mathbb{N}}

\newcommand{\cP}{\mathcal{P}}

\newcommand{\cB}{\mathcal{B}}
\newcommand{\cC}{\mathcal{C}}
\newcommand{\cD}{\mathcal{D}}

\newcommand{\cN}{\mathcal{N}}

\newcommand{\cX}{\mathcal{X}}
\newcommand{\cY}{\mathcal{Y}}

\begin{document}

\title{Deterministic identification for Bernoulli channels and related channels with continuous input}

\author[1]{Pau Colomer\textsuperscript{$\,\star$,}}
\author[2,3]{Christian Deppe\textsuperscript{$\,\dagger$,}}
\author[1,3,4]{Holger Boche\textsuperscript{$\,\ddagger$,}}
\author[5,6]{Andreas Winter\textsuperscript{$\,\S$,}}

\affil[1]{\small\it Chair of Theoretical Information Technology, School of Computation, Information and Technology, \protect\\ Technische Universit\"at M\"unchen, Theresienstra{\ss}e 90, 80333 M\"unchen, Germany.\vspace{2pt}}

\affil[2]{\small\it Institute for Communications Technology, Faculty of Electrical Engineering, Information Technology, \protect\\and Physics, Technische Universit\"at Braunschweig, Schleinitzstraße 22, 38106 Braunschweig, Germany.\vspace{2pt}}

\affil[3]{\small\it 6G-life, 6G research hub, Schleinitzstra{\ss}e 22, 38106 Braunschweig, and Arcisstra{\ss}e 21, 80333 M\"unchen, Germany.\vspace{2pt}}

\affil[4]{\small\it Cluster of Excellence, CeTI, Technische Universität Dresden, Georg-Schumann-Stra{\ss}e 9, 01187 Dresden, Germany.\vspace{2pt}}

\affil[5]{\small\it Chair of Quantum Information and Computing, Department Mathematik/Informatik \protect\\ Abteilung Informatik, Universit\"at zu K\"oln, Weyertal 121, 50931 Köln, Germany.\vspace{2pt}}

\affil[6]{\small\it ICREA \&{} Grup d'Informaci\'o Qu\`antica, Departament de F\'isica, \protect\\ Universitat Aut\`onoma de Barcelona, 08193 Bellaterra (Barcelona), Spain.}

\date{18 February 2026} 

\maketitle

\pagestyle{plain}

\begin{abstract}
For memoryless channels with continuous input alphabets, deterministic identification (DI) typically exhibits a linearithmic ($n\log n$) message growth. However, the exact DI capacity has long remained open due to a persistent gap between the best known achievability and converse bounds. This gap was recently closed for AWGN channels via a novel code construction optimising the ``galaxy'' codes. Here, we extend this approach to the Bernoulli channel and subsequently to any channel $W$ whose image contains a continuous curve of output probability distributions, and hence admits a reduction to the Bernoulli channel restricted to a subinterval of inputs. As a consequence, we prove that the converse bound is tight and establish $\dot{C}_{\text{DI}}(W) = \frac 12$ for this broad class of channels, thereby closing the long-standing capacity gap. 
A similar gap was also observed for the DI rate-reliability tradeoff. We analyse the tradeoff between rate and error of the proposed code and derive improved lower bounds on the reliability function, approaching the converse at leading order in the regime of small error exponents.
\end{abstract}


\blfootnote{
{\textsuperscript{$\star$}pau.colomer@tum.de, \hspace{20pt} \textsuperscript{$\dagger$}christian.deppe@tu-bs.de,}\vspace{1pt}

{\textsuperscript{$\ddagger$}boche@tum.de, \hspace{45pt}\textsuperscript{$\S$}andreas.winter@uni-koeln.de.}}

\section{Introduction}
Identification is a post-Shannon communication paradigm in which a sender wishes to convey a message $i\in[N]$, while the receiver, rather than attempting to fully decode $i$ as in Shannon’s classical transmission setting \cite{Shannon:TheoryCommunication}, only tests whether the transmitted message coincides with a particular message $j\in[N]$ of its choice. In other words, the decoding task is reduced to a binary decision problem: determine reliably whether $i=j$ or $i\neq j$.

This protocol was first proposed by Ahlswede and Dueck in \cite{AD:ID_ViaChannels}. They showed that the length of the identifiable messages scaled exponentially faster with the block length $n$ than the length of the transmittable messages in Shannon's paradigm. Indeed, while the $M$ decodable messages satisfy $\log M\sim Rn$, with $R$ the communication rate, the $N$ identifiable messages satisfy $\log N\sim 2^{Rn}$ \cite{AD:ID_ViaChannels,HanVerdu:ID}. This surprising exponential improvement motivates the relevance of identification for future large-scale network systems \cite{6G_Book}, specially in the field of event-driven or goal-oriented communications \cite{Goal_Oriented_Coms}.

It was already observed in \cite{AD:ID_ViaChannels,AC:DI} that the exponential improvement described above is intrinsically linked to the use of randomness at the encoder. In many applications, however, randomness may be costly, difficult to generate reliably, or undesirable due to implementation constraints. This motivates the study of identification schemes without randomization, i.e., \emph{deterministic identification (DI)}.

DI was first investigated for discrete memoryless channels (DMCs) \cite{AC:DI,SPBD:DI_power}. In that setting, the exponential identification advantage disappears: DI codes can only support a linear growth of the message set, i.e.,
$\log N_{\text{DMC}}\sim nR'$,
although the achievable DI rates $R'$ may exceed the corresponding transmission capacity \cite{SPBD:DI_power}. 
Surprisingly, the situation changes fundamentally when the channel input alphabet is continuous. Indeed, for the Gaussian channel $G$ and the Poisson channel $P$ (under peak/average power constraints), it was shown that the size of the message set under DI can grow as
\(
\log N_{G,P} \sim R' n\log n,
\)
see \cite{SPBD:DI_power, DI-fading, RRGauss-arXiv, DI-poisson_mc}. Following terminology from computer science, we refer to this $n\log n$ scaling as \emph{linearithmic}.

The linearithmic scaling is not limited to the specific channels above, but is rather a prevalent feature of channels with continuous input. Indeed, in \cite{CDBW:DI_classical,CDBW:Reliability-TCOM}, it was shown that for the family of channels with continuous input and discrete output alphabets linearithmic DI codes also exist. 

\begin{definition}\label{def:DI_code}
An $(n,N,\lambda_1,\lambda_2)$ DI code over $n$ uses of a channel $W$ is a family of pairs $\{(u_i,\cD_i) : i\in[N]\}$ with $u_i\in\cX^n$ input sequences and $\cD_i\subset\cY^n$ output regions such that, for all distinct $i, j\in [N]$, and $\lambda_1,\lambda_2>0$:
\[
  W^n(\cD_i|u_i)\ge 1-\lambda_1 \quad\text{and}\quad W^n(\cD_j|u_i)\le\lambda_2.
\]
\end{definition}

The definition highlights another fundamental distinction between Shannon-style message transmission and identification. While in transmission the receiver produces an estimate $\hat{i}$ of the transmitted message $i$, and an error occurs whenever $\hat{i}\neq i$, in identification we have two different types of errors, analogous to the two error events in binary hypothesis testing. The first type is a \emph{missed identification} and it happens when the messages sent and tested are the same but the test rejects the true message. The second type is a \emph{false identification} and it occurs when the messages sent and tested are different but the test incorrectly accepts the wrong message. The parameters $\lambda_1$ and $\lambda_2$ bound respectively the two errors. A good $(n,N,\lambda_1,\lambda_2)$ DI code must control both types of errors simultaneously.
A number $\dot R>0$ is said to be an \emph{achievable linearithmic rate} if there exists a sequence of reliable
$(n,N^{(n)},\lambda_1^{(n)},\lambda_2^{(n)})$ 
DI codes such that
\begin{equation}\label{eq:achievable_rate_intro}
\lambda_1^{(n)},\lambda_2^{(n)} \rightarrow  0,
\,\,\, \text{and}\,\,\,
\dot R(n):=\frac{\log N^{(n)}}{n\log n} \rightarrow \dot{R},
\end{equation}
as $n\to \infty$.
The (linearithmic) DI capacity of a channel is its largest achievable rate:
\begin{equation}\label{eq:def_capacity}
\dot{C}_{\mathrm{DI}}(W)=
\max\bigl\{\dot{R}\,:\,\dot{R}\text{ is achievable}\bigr\}.
\end{equation}
We use the dot on top of rates and capacities to indicate that they are calculated in the linearithmic scale.

A central open problem in DI has been the persistent gap between the best known lower and upper bounds on $\dot C_{\text{DI}}$ for channels with continuous input alphabets. In particular, for additive white Gaussian noise (AWGN) channels $G$ \cite{SPBD:DI_power,DI-fading}, the Poisson channel $P$ \cite{DI-poisson_mc}, and the Bernoulli channel $B$ \cite{CDBW:DI_classical}, the capacity was bounded by
\begin{equation}\label{eq:gap_intro}
    \frac14 \leq \dot C_{\text{DI}}(G), \dot C_{\text{DI}}(P), \dot C_{\text{DI}}(B) \leq \frac12.
\end{equation}
Furthermore, a similar gap was found for the general family of arbitrary channels $W$ with discrete output. To be precise, let $W:\cX\to\cY$ be a channel with input alphabet $\cX$ and output $\cY$, and denote its channel image by
\begin{equation}
W(\cX):=\{W(\cdot|x): x\in\cX\}\subseteq\cP(\cY),
\end{equation}
where $\cP(\cY)$ is the space of probability measures on $\cY$. Also, let $d_M$ denote the Minkowski dimension \cite{Fraser:dimensions,Robinson:dimensions} of this output set $W(\cX)$. Then, it is shown in \cite{CDBW:DI_classical}, that
\begin{equation}\label{eq:gap2_intro}
    \frac{d_M}{4} \leq \dot C_{\text{DI}}(W)\leq  \frac{d_M}{2}.
\end{equation}
Only in a few special cases, where the structure of optimal decoders is explicit, the capacity was known \cite{CDBW:DI_classical,qhtl-arXiv}. 

Recently, an optimised DI code construction, inspired by \cite{galaxy-codes}, was shown to achieve the upper bound for the additive white Gaussian noise (AWGN) channel $G$, thereby establishing $\dot C_{\text{DI}}(G)=\frac12$ and closing the Gaussian capacity gap \cite{CDBW:Gauss}. This naturally raises the question of whether the upper bound in Equation~\eqref{eq:gap_intro} is also achievable for the Poisson channel, the Bernoulli channel, and broader classes of continuous-input discrete-output channels.

In this paper, we answer this question in the affirmative. Building on the underlying geometric ideas of the Gaussian construction in \cite{CDBW:Gauss} (in a non-trivial and somewhat counter-intuitive way), we propose a new DI coding scheme for the Bernoulli channel $B$ (see the channel model in Section \ref{sec:channel_model}) and prove that it achieves the converse bound in Equation~\eqref{eq:gap_intro}, thereby establishing that $\dot C_{\text{DI}}(B)=\frac12$ (see Section \ref{ssec:opt_code}). 

Then, in Section~\ref{ssec:restricted_bern}, we extend the construction to the restricted Bernoulli channel and, via the reducibility argument Lemma~\ref{lemma:B_to_W}, to any channel $Q$ whose image is the union of a finite number of continuous curves and a residual set of (upper) Minkowski dimension at most $d_M=1$, which include the AWGN and Poisson cases. For all these, we obtain an achievable linearithmic rate of $\frac12$ which, combined with the converse bounds in Equations \eqref{eq:gap_intro} and \eqref{eq:gap2_intro}, allows us to establish for the first time the exact capacities
\begin{equation}
\dot C_{\text{DI}}(P)=\dot C_{\text{DI}}(B)=\dot C_{\text{DI}}(Q)=\frac12,
\end{equation}
thereby closing the capacity gap for this broad class of channels.
The results also reproduce the main capacity result in \cite{CDBW:Gauss} for AWGN channels: $\dot C_{\text{DI}}(G)=\frac12$.

For the Poisson channel, we additionally include an explicit capacity-achieving construction in Section \ref{sec:poisson}. 

In Section \ref{sec:RR}, we develop a rate-reliability study of the proposed code, observing that, in the regime of small error exponents, the lower bound on the rate-reliability function matches the known upper bound at leading order. We conclude in Section \ref{sec:conclusions} with a final discussion. 

These results close the long-standing capacity gap for DI in the linearithmic regime for a broad class of continuous-input channels, and yield improved achievability bounds for the associated reliability functions which are tight at first order.


\section{Channel model}\label{sec:channel_model}
We start by considering communication over the \emph{Bernoulli channel} $B$. Given an input $x\in[0,1]$, this channel outputs a binary symbol $y\in\{0,1\}$ according to the Bernoulli distribution with parameter $x$:
\begin{equation}\label{eq:BernoulliChannel}
B(y|x) = B_x(y) 
= \begin{cases}
    x   & \text{if }\, y=1, \\
    1-x & \text{if }\, y=0.
  \end{cases}
\end{equation}
On block length $n$ we have a continuous set of inputs $x^n = x_1\dots x_n \in[0;1]^n$ which are points in the unit hypercube of dimension $n$. On the other hand, we have a finite set of $2^n$ possible outputs $y^n\in\{0,1\}^n$ (the vertices of the afore-mentioned cube). 

The Bernoulli channel is a significant example not only in its own right, but also because it allows the construction of DI codes for other channels. In particular, restricting its input alphabet to an interval $\cX=[a,b]\subset[0,1]$, with $0\leq a < b \leq 1$, yields the \emph{restricted Bernoulli channel}, denoted $B|_{[a,b]}$, which we will show has linearithmic DI capacity $\dot C_{\text{DI}}(B|_{[a,b]})=\frac12$. As observed in \cite{CDBW:DI_classical}, broad classes of channels can be reduced to this restricted model via suitable classical pre- and post-processing.

To formalise this reduction principle, assume that the channel's image $W(\cX)$ contains a non-trivial continuous curve of output distributions, we call these channels \emph{Bernoulli-reducible}. Then one may choose a binary-valued post-processing channel $V:\cY\to\{0,1\}$ such that the reduced channel
\[
W':=V\circ W:\cX\to\{0,1\}
\]
is non-constant. Since $W'$ has binary output, its channel image can be identified with a subset of $[0,1]$ via the Bernoulli parameter $p(x):=W'(1|x)$. By continuity, the set $\{p(x):x\in\cX\}$ contains a non-trivial interval $[a,b]\subset[0,1]$, and hence $W'$ simulates the restricted Bernoulli channel $B|_{[a,b]}$. Applying any DI code for $B|_{[a,b]}$ to $W'$ and concatenating with the post-processing $V$ yields a DI code for $W$. This leads to the following consequence.

\begin{lemma}[cf.\ Sec.~V.A in \cite{CDBW:DI_classical}]\label{lemma:B_to_W}
If $W(\cX)$ contains a continuous curve, then there exists an interval $[a,b]\subset[0,1]$ such that $W$ can simulate the restricted Bernoulli channel $B|_{[a,b]}$ via suitable classical pre- and post-processing. Consequently,
\[
\dot{C}_{\text{DI}}(W) \geq \dot{C}_{\text{DI}}(B|_{[a,b]}).
\]
\end{lemma}

Notice the power of this reduction: if we construct a DI code achieving the upper bound in Equation \eqref{eq:gap_intro} for the restricted Bernoulli channel, i.e.\ $\dot{C}_{\text{DI}}(B|_{[a,b]})\geq \frac12$, then Lemma~\ref{lemma:B_to_W} implies the achievability bound $\dot{C}_{\text{DI}}(W)\geq \frac12$ for every channel whose image contains a non-trivial continuous curve, i.e., every Bernoulli-reducible channel. This family includes continuous-input and discrete-output channels $Q$ with $d_M=1$, as well as continuous-output models such as the Gaussian channel $G$ and counting-output models such as the Poisson channel $P$ (under standard input constraints). Since in these settings a matching converse bound is available, the reduction lemma establishes their exact DI capacity.


\section{Preliminaries and strategy intuition}
The first DI code constructions were based on typicality arguments, in particular on the \emph{hypothesis testing lemma} \cite{CDBW:DI_classical}. This lemma shows that if one can find code words $\{u_i\}_{i=1}^N$ whose induced output distributions $W_{u_i}$ are sufficiently separated, then a good DI code can be obtained by choosing decoding regions as entropy-typical sets \cite{CK:book2011}. Concretely, for $\delta>0$ one defines
\begin{equation}\label{eq:entropy-typical-set}
\mathcal{T}_{u_i}^\delta \! := \left\{ y^n\in\mathcal{Y}^n \!: \left| \log W_{u_i}(y^n) \!+\! H(W_{u_i}) \right| \leq \delta\sqrt{n} \right\},
\end{equation}
where $H(W_{u_i})$ denotes the Shannon entropy of the output distribution induced by the input sequence $u_i$.

For the Bernoulli channel, the distance needed between output probability distributions translates directly to a minimum Euclidean distance between the corresponding input sequences. In particular, one can construct DI codes by taking the centres of an Euclidean ball packing of radius on the order of $n^{1/4}$ in $[0,1]^n$, together with their associated typical decoding regions. Analysing such a construction yields the lower bound $\dot C_{\text{DI}}(B)\geq \frac14$ appearing in \eqref{eq:gap_intro}.

The typicality-based approach above, while very general, does not achieve higher linearithmic rates. However, a new code construction was recently proposed for DI over AWGN channels based on a purely geometric technique, moving away from typicality arguments \cite{CDBW:Gauss}, which improved the rates obtained with the previous strategy. The cornerstone property there is that the (high-dimensional) Gaussian noise is, with high probability, nearly orthogonal to any fixed direction. Using a carefully designed multi-layer geometric arrangement of code words, one obtains explicit DI codes achieving the upper bound, and hence $\dot C_{\text{DI}}(G)=\frac12$.

At first glance, one might expect that the near-orthogonality property exploited in \cite{CDBW:Gauss} is specific to the AWGN setting, since Gaussian noise is independent of the input and isotropic (distributed with the same probability in all directions). In contrast, the noise in the Bernoulli channel depends strongly on the input sequence: the output always lies on a vertex of the hypercube, and vertices closer to the input are more likely than those farther away. Nevertheless, a form of the same concentration phenomenon persists: the projection of the Bernoulli noise (in high-dimension settings) onto any fixed direction is unlikely to be large.

Let us start by writing the noise vector of a Bernoulli channel in an additive way. For each input letter $x_i\in[0;1]$, the output can be written as the following random variable:
\begin{equation}
Y_i=x_i+Z_i, 
\,\,\text{where}\,\,
Z_i=\begin{cases}
-x_i    &\text{with prob. } 1-x_i, \\
1-x_i  &\text{with prob. } x_i.  
\end{cases}
\end{equation}
Then, for any input sequence $x^n=x_1,\dots,x_n$, the output can be written as $Y^n=x^n+Z^n$ with $Z^n=Z_1,\dots,Z_n$. Let us express the noise as the vector $\vec Z=(Z_1,\dots,Z_n)$.

\begin{proposition}\label{proposition:projective_property}
Let $\vec{v}\in\RR^n$ be a vector with $\|\vec v\|>0$, and let $\vec{Z}=(Z_1,\dots,Z_n)$ denote the additive noise induced by the Bernoulli channel for some input sequence $x^n\in[0;1]^n$. Then, for all $t>0$, the orthogonal projection $\Pi_{v}\vec Z$ of the noise vector onto $\text{span}(\vec{v})$, satisfies:
\[
\Pr\left\{\|\Pi_{v}\vec Z\|>t\right\}\leq 2\exp \left(-2t^2\right).
\]
\end{proposition}
\begin{proof}
Take $\vec v$ in any direction, and let $\vec e:=\vec v/\|\vec v\|$ be its normalization, i.e., $\|\vec e\|^2=\sum_{i=1}^n e_i^2=1$. Then, the projection of the noise $\vec Z$ onto $\text{span}(\vec v)$ [or, equivalently, $\text{span}(\vec e)$] is
\begin{equation}\label{eq:orth_projection}
\Pi_{v}\vec Z:=\frac{\langle\vec Z,\vec v\rangle}{\|\vec v\|^2}\vec v=\langle\vec Z,\vec e\rangle\vec e.
\end{equation}

Hoeffding's inequality \cite{Hoeffding_inequality} states that given a sum of independent random variables $X^n=X_1+\dots+X_n$, with each $X_i\in[a_i,b_i]$, then for all $t>0$
\begin{equation}\label{eq:Hoef_bound}
\!\!\Pr\left\{|X^n-\EE[X^n]|\!\geq\! t\right\}\leq 2\exp\left(-\frac{2t^2}{\sum_{i=1}^n \left|b_i-a_i\right|^2}\right)\!.
\end{equation}

Let  
\(
S_v:=\langle\vec Z,\vec e\rangle=\sum_{i=1}^n Z_ie_i=:\sum_{i=1}^n S_i
\)
be the normalized inner product. Observe that $S_v$ is a sum of independent random variables $S_i$, which satisfy $S_i\in[-x_i|e_i|,(1-x_i)|e_i|]$. Then, noticing first that the mean is null
\begin{equation}
    \EE[S_v]=\sum_{i=1}^n x_i(1-x_i)e_i-(1-x_i)x_ie_i=0,
\end{equation}
and that 
\begin{equation}
\begin{split}
\sum_{i=1}^n(b_i-a_i)^2&=\sum_{i=1}^n \left[(1-x_i)|e_i|-(-x_i)|e_i|\right]^2\\
&=\sum_{i=1}^n |e_i|^2=1,
\end{split}
\end{equation}
we can apply Hoeffding's bound in Equation~\eqref{eq:Hoef_bound} obtaining, for all $t>0$,
\begin{equation}
\Pr\left\{\|\Pi_{v}\vec Z\|>t\right\}=\Pr\left\{|S_v|>t\right\}\leq 2\exp \left(-2t^2\right),
\end{equation}
where the first inequality follows from $\|\Pi_{ v}\vec Z\|\!=\!\left\|S_v\vec e\right\|\!=\!|S_v|$. The proof is complete.
\end{proof}

The proposition shows that although the Bernoulli noise depends on the input sequence and is highly non-isotropic (some directions are much more probable than others), its projection onto any fixed direction still concentrates sharply around zero. In other words, the component of the noise along a prescribed direction is typically small. In the next section we show that this concentration property alone suffices to construct a multi-layer geometric DI code for the Bernoulli channel, in the spirit of \cite{CDBW:Gauss}.


\section{Optimal DI code construction for the Bernoulli channel}\label{sec:optimal_bern}
We show in this section that a modification of the universal DI code for Gaussian channels proposed in \cite{CDBW:Gauss} can also be applied to the Bernoulli channel $B$. We begin in Section~\ref{ssec:primitive} with a primitive construction inspired by the Gaussian scheme, followed by an error analysis in Section~\ref{ssec:error_analysis}. In Section~\ref{ssec:opt_code} we show that an expurgation of the primitive construction yields a reliable DI code for the Bernoulli channel. Moreover, Section~\ref{ssec:restricted_bern} shows that the same construction extends to any restricted Bernoulli channel $B|_{[a,b]}$ and, by Lemma~\ref{lemma:B_to_W}, to all Bernoulli-reducible channels. Together with the converse results stated in the introduction, these results allow us to establish the exact capacity value for any channel $Q$ with continuous input and discrete output with $d_M=1$ (containing the Bernoulli and restricted-Bernoulli channels), and for the Poisson $P$ and Gaussian channels $G$, finally closing the capacity gap in the Poisson case, and reproducing the recent results in \cite{CDBW:Gauss} for the Gaussian.

\subsection{Primitive code book description}\label{ssec:primitive}
In this section we construct a primitive code book for the Bernoulli channel. We call it \emph{primitive} because not all code words satisfy the power constraint, meaning that some of them are unusable for DI over the Bernoulli channel. In Section \ref{ssec:opt_code}, an expurgation of the construction proposed here is performed to produce a true code for the Bernoulli channel (with all code words inside the available input space).

The primitive code book is a multi-layer geometric arrangement of points on the surface of a high-dimensional sphere equivalent to the universal construction for AWGN channels in \cite{CDBW:Gauss}. We construct it sequentially through a layered system satisfying the following minimum projective separation condition.

\begin{definition}\label{def:angle-dense}
Given $d>0$, a set of points $\{u_j\}_{j=1}^{N}$ on the surface of an $n$-dimensional ball $\cB(r,n)$ of radius $r$ is called $d$-\emph{angle-dense} if for every pair of distinct points $u_j\neq u_k$ it holds that
\[
\|\Pi_{\vec{u}_k}\vec{u}_j-\vec{u}_k\|\geq d.
\]
\end{definition}

Let $\vec o_{s_0}=(\frac12,\dots,\frac12)$ be the centre of the construction (note that it is also the centre of the $n$-dimensional cube defining the input space). Then, let $\{\vec o_{s_1}\}_{s_1=1}^{N_1}$ be the $N_1$ points in an $d$-angle-dense arrangement in the $n$-dimensional ball $\cB_{s_0}(r_1,n)$ of radius $r_1$ centred at $\vec o_{s_0}$. This spherical arrangement constitutes the \emph{first layer} of the code. 

Around each element $\vec o_{s_1}$ of the first layer, we construct a \emph{second layer} system, consisting of another $d$-angle-dense arrangement of points $\{\vec o_{s_1,s_2}\}_{s_2=1}^{N_2}$ on the surface of the ball 
\(
\cB_{s_1}(r_2,n-1)\in\text{span}(\vec v_{s_0\to s_1}^\perp),
\)
where $\vec v_{s_o\to s_1}:=\vec o_{s_1}-\vec o_{s_0}$. In other words, the second layer arrangement is constructed in the subspace orthogonal to the vector connecting the centre of the code $\vec o_{s_0}$ and the first-layer element $\vec o_{s_1}$, which is the centre of the corresponding second-layer system, thereby losing one dimension.

Around each second-layer element, a further $d$-angle-dense system is added, forming a third layer, and so on. At layer $\ell$, we have elements $\{\vec o_{s^\ell}\}_{s_1,\dots,s_\ell=1}^{N_1,\dots, N_\ell}$ uniquely defined by the string $s^\ell=s_1,\dots,s_\ell$ that denotes the particular point at each layer around which the element $\vec o_{s^\ell}$ can be found. In other words, defining the vectors $\vec v_{s_{\ell-1}\to s_\ell}:=\vec o_{s^\ell}-\vec o_{s^{\ell-1}}$ at each layer $\ell$, we have that 
\(
    \vec o_{s^\ell}=\sum_{i=1}^\ell\vec v_{s_{i-1}\to s_{i}}.
\)

Around each $\ell$-layer element $\vec o_{s^\ell}$, another $d$-angle-dense arrangement of points $\{\vec o_{s^{\ell+1}}\}_{s_{\ell+1}=1}^{N_{\ell+1}}$ on the surface of the ball 
\(
\cB_{s^\ell}(r_\ell,n-\ell+1)\in\text{span}(\vec v_{s_0\to s_1},\dots,\vec v_{s_{\ell-1}\to s_\ell})^\perp
\)
is added. See Figure \ref{fig:subspaces} for visual intuition on the notation. 

\begin{figure}
    \centering
    \includegraphics[width=0.95\linewidth]{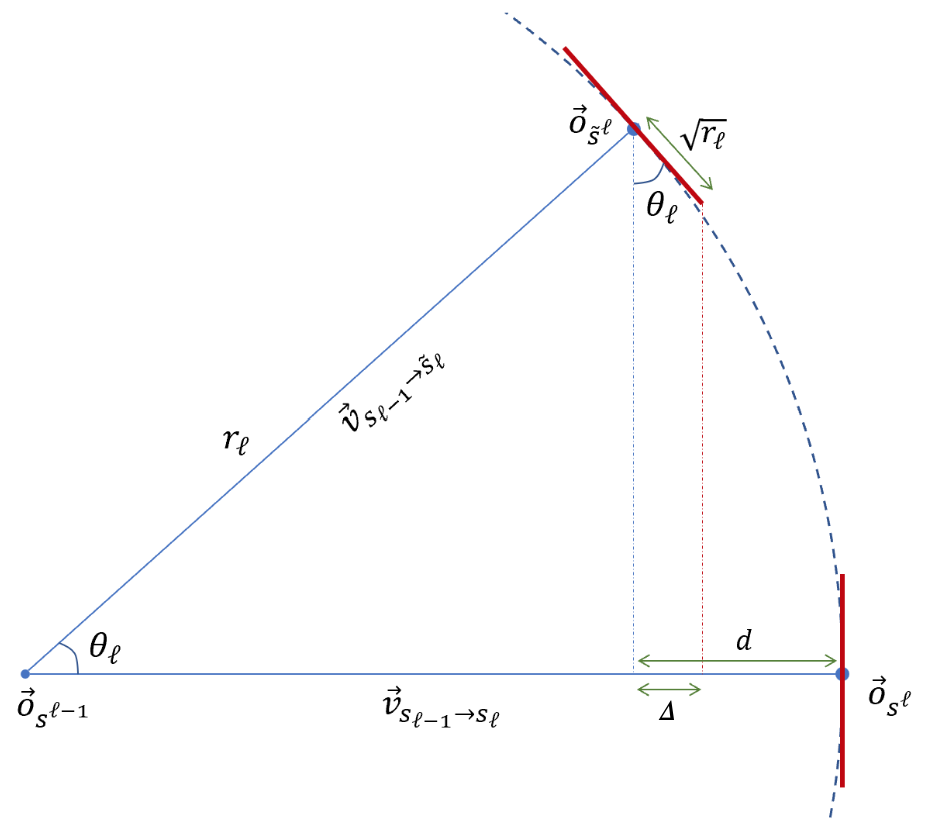}
    \caption{Schematic structure of a $d$-angle-dense arrangement at layer $\ell$ around the point $\vec o_{s^{\ell-1}}$ (which, at the same time, is a point in the $\ell-1$-layer). Observe that the orthogonal projection of the $\ell$ layer point $\vec o_{\tilde s^\ell}$ onto the vector $\vec v_{s_{\ell-1}\to s_\ell}$ is at least $d$. The two red lines, respectively orthogonal to the vectors $\vec v_{s_{\ell-1}\to s_\ell}$ and $\vec v_{s_{\ell-1}\to \tilde s_\ell}$ represent the regions where the $\ell+1$ layer will be arranged. Notice that any point in the $\ell+1$ layer around $\vec o_{\tilde s^\ell}$ will have an orthogonal projection onto $\vec v_{s_{\ell-1}\to s_\ell}$ at a distance at most $\Delta$ from $\Pi_{s_{\ell-1}\to s_\ell}\vec o_{\tilde s^\ell}$ (see Proposition~\ref{proposition:concentration}).}
    \label{fig:subspaces}
\end{figure}

\begin{definition}\label{def:primitive_code book}
Given a fixed number of layers $L$ (independent of $n$), the \emph{primitive code book} is the set of $N_P=\prod_{\ell=1}^LN_\ell$ points in the innermost layer of the construction above:
\[
\left\{\vec o_{s^L}\in\mathrm{span}(\vec v_{s_0\to s_1},\dots,\vec v_{s_{L-1}\to s_L})^\perp\right\}_{s_1,\dots,s_L=1}^{N_1,\dots,N_L}.
\]
\end{definition}

We will first study this code book under the following choices of radius at each layer $r_\ell=n^{(1-\delta)/2^\ell}$, given some small parameter $\delta>0$, and minimum projective distance in the dense arrangements $d:=3\log n$.

\begin{proposition}\label{proposition:primitive_code_size}
The number of elements in a particular $\ell$-layer $d$-angle dense arrangement is bounded by
\[
\log N_\ell\geq\frac{n-\ell+1}{2}\log\left(\frac{2r_\ell}{d}\right).
\]
Then, the total size $N_P$ of a primitive code book for DI over the Bernoulli channel can be bounded by
\[
\log N_P\geq\frac{n\log n}{2}\sum_{\ell=1}^L\frac{1-\delta}{2^\ell}-O(n\log\log n).
\]

\end{proposition}
\begin{proof}
It is a standard consequence of spherical packing bounds (see, e.g., \cite{packing_book}) that the number $N$ of points on the surface of a $D$-dimensional hypersphere at a minimum angular separation $\theta$ can be lower bounded by:
\begin{equation}
    N\geq\left(\sin\frac\theta2\right)^{-D[1+o(1)]}.
\end{equation}
Notice that a $d$-angle dense arrangement at a layer $\ell$, i.e. an arrangement of points on the surface of the ball $\cB(r_\ell,n-\ell+1)$ such that the projective property in Definition~\ref{def:angle-dense} holds, must have a minimum angular separation between their elements satisfying $r_\ell(1-\cos\theta_\ell) \geq d$, which implies 
\begin{equation}\label{eq:angles}
    \sin\frac{\theta_\ell}{2}\geq\sqrt{\frac{d}{2 r_\ell}}.
\end{equation}
Then, putting the last two equations together and noticing that at layer $\ell$ the dimension $D=n-\ell+1$ we observe
\begin{equation}
\log N_\ell\geq\frac{n-\ell+1}{2}\log\left(\frac{2r_\ell}{d}\right).
\end{equation}

Introduce now the values $r_\ell=n^{(1-\delta)/2^\ell}$ and $d=3\log n$ that we have previously chosen for the primitive code book. Then, since $N_P=\prod_{\ell=1}^LN_\ell$, one gets
\begin{equation*}
\log N_P\geq\frac{n\log n}{2}\sum_{\ell=1}^L\frac{1-\delta}{2^\ell}-O\left(n\log\log n\right).\qedhere
\end{equation*}
\end{proof}

The complex geometrical structure of the code book above, specifically the construction of each layer in the particular orthogonal subspaces we have chosen, force the projections of elements in inner layers to concentrate around the projections of their corresponding centres, see Figure~\ref{fig:subspaces} for visual intuition. This property is proved in \cite[Lemma~11]{CDBW:Gauss} and sketched in the proposition below:

\begin{proposition}\label{proposition:concentration}
Let $\vec o_{s^L}\neq\vec o_{\tilde s^L}$ be two different last-layer elements (primitive code words) which differ at some layer $\ell$, that is, $s_\ell\in s^L$ is different from $\tilde s_\ell\in \tilde s^L$. Then,
\[
\|\Pi_{s_{\ell-1}\to s_\ell}\vec o_{\tilde s^L}-\Pi_{s_{\ell-1}\to s_\ell}\vec o_{s^L}\|\geq d-\Delta\geq d-\sqrt{2Ld}.
\]    
\end{proposition}
\begin{proof}
One just needs to notice that all the vectors $\vec v_{s_{\ell-1}\to s_\ell}$ corresponding to the same word, say $\vec o_{\tilde s^L}$ are mutually orthogonal. Hence, by Pythagoras, the distance between the code word $\vec o_{\tilde s^L}$ and its arrangement centre $\vec o_{\tilde s^\ell}$ at layer $\ell$ satisfies
\begin{equation}
\|\vec o_{\tilde s^L}-\vec o_{\tilde s^\ell}\|^2
\leq\sum_{j=\ell}^L \|\vec v_{s_{j-1}\to s_j}\|^2
=\sum_{j=\ell}^L r_j^2.
\end{equation}
Then, its projection onto the vector $\vec v_{s_{\ell-1}\to s_\ell}$ can be separated from the projection $\Pi_{s_{\ell-1}\to s_\ell}\vec o_{\tilde s^\ell}$ by at most a distance (cf.~Figure~\ref{fig:subspaces})
\begin{equation}\label{eq:delta_def}
    \Delta:=\sin\theta_\ell\sqrt{\sum_{j=\ell+1}^L r_j^2}.
\end{equation}

In the worst case scenario, the two code word projections we are comparing differ already at $\ell=1$, and angles are the smallest possible, that is, $\cos\theta_\ell=1-d/r_\ell$, and thus $\sin\theta_\ell=\sqrt{\frac{2d}{r_\ell}-\frac{d^2}{r_\ell^2}}$. Then,
\begin{equation}
  \Delta\leq\sqrt{L}r_2\sin\theta_1=\sqrt{L\left(2d-\frac{d^2}{r_1}\right)}\leq\sqrt{2Ld}.
\end{equation}
Where in the second line we used that $r_2=\sqrt{r_1}$. Then, the proposition claim follows trivially.
\end{proof}

With our choice of the parameter $d=3\log n$ and noticing again that $\Pi_{s_{\ell-1}\to s_\ell}\vec o_{s^L}=\vec o_{s^\ell}$ one immediately observes that, for sufficiently large $n$, the following holds:
\begin{equation}\label{eq:concentration_property}
\|\Pi_{s_{\ell-1}\to s_\ell}\vec o_{\tilde s^L}-\vec o_{s^\ell}\|\geq 3\log n -O(\sqrt{\log n})\geq 2\log n.
\end{equation}

We have now all the necessary tools to show that the primitive code can be reliable. For more details on the primitive code construction we refer the reader to \cite[Sections II \& III]{CDBW:Gauss}, which include further geometric intuition and a more detailed step-by-step construction.

\subsection{Decoding strategy and error analysis}\label{ssec:error_analysis}
While the primitive code word arrangement is equivalent to that in \cite{CDBW:Gauss}, we must bound the errors differently than as in the Gaussian case due to the different nature of the noise. 
In this section, we will describe the decoding strategy and show that the elements of the primitive code book, described in the previous section, can be identified reliably, with vanishing missed and false identification error probabilities. 

Let us use a sequential decoding strategy. For each possible $i\in[N]$ encoded into the primitive code book word $\vec o_{s^L}$, the following decoding effect is defined at each layer $\ell$:
\begin{equation}\label{eq:layer-decoder}
\cD_{i}^{(\ell)}=\{\vec{y}\in\RR^n:\|\Pi_{s_{\ell-1}\to s_{\ell}}\vec{y}-\vec{o}_{s^\ell}\|\leq\log n\},
\end{equation}
where $\Pi_{s_{\ell-1}\to s_{\ell}}$ denotes the orthogonal projection onto the $\text{span}(\vec v_{s_{\ell-1}\to s_{\ell}})$. 

To identify the message $i$, the test above must produce a positive answer for that decoder in each layer. In other words, the overall identification test is defined as $\cD_i=\cap_{\ell=1}^L\cD_{i}^{(\ell)}$. On the other hand, as soon as a test at a single layer fails, the message is discarded.

If the same primitive code word $\vec o_{s^L}$ (corresponding to message $i$) is sent and tested, the probability that the identification test at layer $\ell$ in Equation~\eqref{eq:layer-decoder} produces a positive outcome is given by:
\begin{align}
\Pr&(\cD_i^{(\ell)}|\vec o_{s^L})=\Pr\left\{\|\Pi_{s_{\ell-1}\to s_{\ell}}\vec{Y}-\vec{o}_{s^\ell}\|\leq\log n|\vec o_{s^L}\right\}\nonumber\\
&\stackrel{(\alpha)}{=}\Pr\left\{\|\Pi_{s_{\ell-1}\to s_{\ell}}\vec o_{s^L}+\Pi_{s_{\ell-1}\to s_{\ell}}\vec{Z}-\vec{o}_{s^\ell}\|\leq \log n\right\}\nonumber\\
&\stackrel{(\beta)}{=}\Pr\left\{\|\Pi_{s_{\ell-1}\to s_{\ell}}\vec{Z}\|\leq\log n\right\}\nonumber\\
&\stackrel{(\gamma)}{\geq}1-2\exp\left[-2(\log n)^2\right]=:1-\lambda,\label{eq:primitive_layer_error1}
\end{align}
where in $(\alpha)$ we have used the additive noise expression of the Bernoulli channel, described in Equation~\eqref{eq:BernoulliChannel}, and the linearity of the projections; for $(\beta)$ one just needs to observe that $\Pi_{s_{\ell-1}\to s_{\ell}}\vec o_{s^L}=\vec{o}_{s^\ell}$, and $(\gamma)$ follows from Proposition~\ref{proposition:projective_property}. In the last equality we have defined
\(
    \lambda:=2\exp\left[-2(\log n)^2\right].
\)

The probability of correct identification at each layer, bounded in Equation~\eqref{eq:primitive_layer_error1}, is the same across the $L$ layers so, by the union bound, the total missed identification probability is bounded for all $i\in[N]$ by:
\begin{equation}\label{eq:main_Pe1}
\begin{split}
\!\!P_{e,1}(i):=\Pr\left(\bigcup_{\ell=1}^{L}\cD_i^{(\ell)^c}| i\right)&\leq\sum_{\ell=1}^{L}\Pr\left(\cD_i^{(\ell)^c}|\vec o_{s^L}\right)\\
&\leq2L\exp\left[-2(\log n)^2\right]\\
&=L\lambda:=\lambda_1,
\end{split}
\end{equation}
where the superscript $c$ indicates the complementary set. Notice that, for any fixed $L$ (independent of $n$) the missed identification probability above vanishes as $n\to \infty$. 

If the message $j$, encoded into the primitive code word $\vec o_{\tilde s^L}$, is sent, and a different message $i\neq j\in[N]$, encoded into $\vec o_{s^L}$, is tested; then, the probability that at a layer $\ell$ in which $s_l\in s^L$ is different from $\tilde s_\ell\in\tilde s^L$ the outcome of the test in Equation~\eqref{eq:layer-decoder} is (wrongly) positive can be bounded as follows:
\begin{align}
\Pr&(\cD_i^{(\ell)}|\vec o_{\tilde s^L})=\Pr\left\{\|\Pi_{s_{\ell-1}\to s_{\ell}}\vec{Y}-\vec{o}_{s^\ell}\|\leq\log n|\vec{o}_{\tilde{s}^{L}}\right\}\nonumber\\
&\!\stackrel{(\alpha)}{=}\Pr\left\{\|\Pi_{s_{\ell-1}\to s_{\ell}}\vec{o}_{\tilde{s}^{L}}+\Pi_{s_{\ell-1}\to s_{\ell}}\vec{Z}-\vec{o}_{s^\ell}\|\leq \log n\right\}\nonumber\\
&\!\stackrel{(\beta)}{\leq}\Pr\left\{\|\Pi_{s_{\ell-1}\to s_{\ell}}\vec{o}_{\tilde{s}^{L}}-\vec{o}_{s^\ell}\|-\|\Pi_{s_{\ell-1}\to s_{\ell}}\vec{Z}\|\leq\log n\right\}\nonumber\\
&\!\stackrel{(\gamma)}{\leq}\Pr\left\{\|\Pi_{s_{\ell-1}\to s_{\ell}}\vec{Z}\|\geq\log n\right\}\nonumber\\
&\!\stackrel{(\eta)}{\leq}2\exp\left[-2(\log n)^2\right]=\lambda,\label{eq:primitive_layer_error2}
\end{align}
where $(\alpha)$ follows from the additive expression of the channel effect $Y^n=x^n+Z^n$, in $(\beta)$ the reverse triangle inequality has been used, $(\gamma)$ follows from Equation~\eqref{eq:concentration_property}, and $(\eta)$ from Proposition~\ref{proposition:projective_property}.

Since the error of the second type can be bounded at all other layers by $1$, we can conclude that the overall false identification error is given, for any two distinct messages $i\neq j\in[N]$, by:
\begin{equation}\label{eq:main_Pe2}
P_{e,2}(i|j)=\Pr\left(\cD_i|j\right)\leq\Pr\left(\cD_i^{(\ell)}| j\right)\leq\lambda:=\lambda_2.
\end{equation}

We have shown that the primitive code described in Section~\ref{ssec:primitive} can be reliably identified by the sequential decoding strategy proposed in this section. Indeed, using the decoding effects in Equation~\eqref{eq:layer-decoder}, both the missed identification and false identification error probabilities can be bounded by objects that vanish as $n\to\infty$ ($\lambda_1=L\lambda$ and $\lambda_2=\lambda$, respectively).

\subsection{Capacity-achieving code for the Bernoulli channel}\label{ssec:opt_code}
The primitive code book described in the previous sections is a set of points arranged on a sphere of radius $R$, with
\begin{equation}
R^2=\sum_{\ell=1}^Lr_\ell^2=\sum_{\ell=1}^L n^\frac{1-\delta}{2^{\ell-1}}=n^{1-\delta}+o(\sqrt{n}),
\end{equation}
but a DI code book for the Bernoulli channel has to be contained inside the available input space (an $n$ dimensional unit cube). A priori, this seems like a serious problem. How can any point on the surface of a sphere with large radius (indeed, $R$ grows with the block length, which we want to analyse in the limit $n\to\infty$) be inside a cube of small and constant side (in fact, the side is 1 for the general Bernoulli channel and $b-a<1$ in the restricted $B|_{[a,b]}$ case)? 

Well, one needs to notice that even if the constant side of a square is small, the vertices $V_k\in\{0,1\}^n$ with $k\in[1,2^n]$, are far away from the cube centre $\vec o_{s_0}=(\frac12,\dots,\frac12)$. Indeed, the distance $D$ from the cube centre to any vertex $V_k$ satisfies 
\begin{equation}
    D^2=\left\|V_k-\left(\frac12,\dots,\frac12\right)\right\|^2=\sum_{t=1}^n\frac14=\frac n4.
\end{equation}

Notice that, for sufficiently large $n$, we have $D>R$. Therefore, near the vertices, the surface of the $n$-dimensional hypersphere of radius $R$ is inside the unit $n$-dimensional hypercube (see Figure~\ref{fig:urchin} for further insight). In this section, we will calculate how many primitive code words are inside the cube and can therefore form a valid DI code for the Bernoulli channel. This will allow us to calculate a lower bound on the rate.

\begin{figure}
    \centering
    \includegraphics[width=0.75\linewidth]{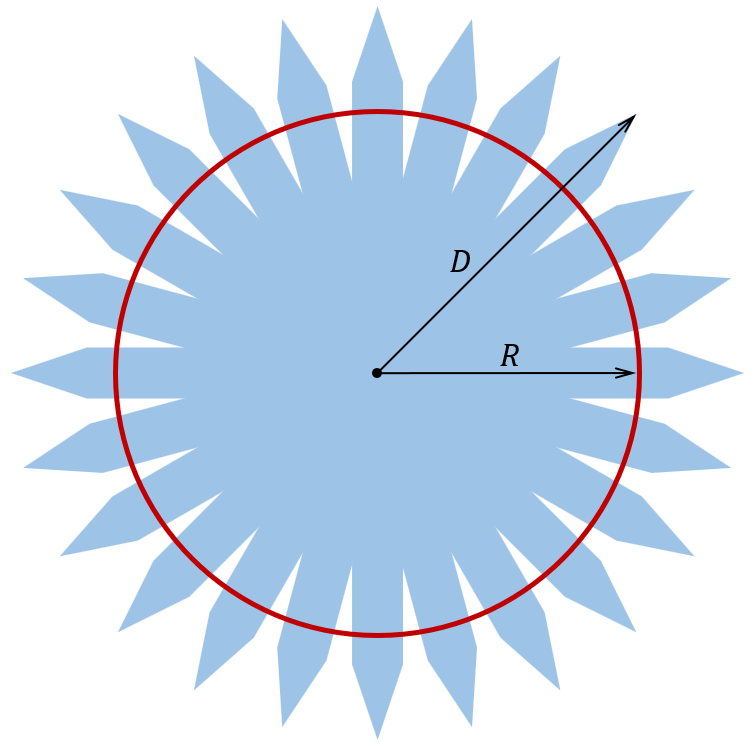}
    \caption{In high dimensions, a cube is better pictured not as a solid box but as a star or sea urchin: a relatively round core (the Gaussian width of a cube is roughly the same as that of its circumscribed ball \cite{Vershynin:high_dimensions_book}) with exponentially many ``spines'' extending far away from the centre towards the $2^n$ vertices. Our code lives on the surface of the sphere of radius $R$ (red line above), which overlaps (is inside) the cube (the blue sea urchin) in most of its volume.}
    \label{fig:urchin}
\end{figure}

Let us start by understanding how a ``typical'' point on the surface of the sphere looks like. A standard device in high-dimensional geometry is to generate a uniform point on the sphere by first sampling a Gaussian vector and then normalising it. Let $T=\frac{R}{\|G\|}(G_1,\dots,G_n)$, where the $G_i\sim\cN(0,1)$ are identical and independent distributed normal Gaussians. Then, $X$ is a uniformly distributed random point on the surface of the sphere $\cB(R,n)$, denoted by $S^{n-1}(R)$. Each marginal distribution $T_i=\frac{R}{\|G\|}G_i$ has sub-Gaussian tails satisfying
\begin{equation}\label{eq:sub-gaussian_tail}
\Pr\left\{|T_i|> t\right\}\leq 2\exp\left(-\frac{cnt^2}{R^2}\right),
\end{equation}
for some absolute constant $c>0$ and all $t>0$. Since, for sufficiently large $n$, we have $R^2=n^{1-\delta}+o(\sqrt{n})\leq2n^{1-\delta}$, and choosing $t=\frac12$, one gets
\begin{equation}\label{eq:out_cube_1Dcondition}
\!\!\Pr\left\{|T_i|> \frac12\right\}\leq 2\exp\left(-\frac{cn}{8n^{1-\delta}}\right)= 2\exp\left(-\frac{cn^\delta}{8}\right)\!.
\end{equation}

Let the set of points on the surface $S^{n-1}(R)$ which are outside a concentric cube of side $1$ be denoted by
\begin{equation}
\mathbf{T}_{\text{out}}:=\left\{t\in S^{n-1}(R):\exists i, |t_i|>\frac12\right\}.
\end{equation}
Then, by the union bound,
\begin{equation}\label{eq:out_cube_Multi_D_condition}
\!\!\Pr\left\{T\in \mathbf{T}_{\text{out}}\right\}\!\leq\!\sum_{i=1}^n\Pr\!\left\{\!|T_i|> \frac12\!\right\}\!\leq 2n\exp\left(-\frac{cn^\delta}{8}\right)\!.
\end{equation}
In other words, the probability that a uniform random point on the surface of the sphere falls outside the cube is very small. In fact, it vanishes with $n$.

Let $\sigma(V):=\text{Area}[V]/\text{Area}[S^{n-1}(R)]$ be the normalised surface measure on $S^{n-1}(R)$. Then, the set of points on the surface of the ball which are also inside the cube, denoted by $\mathbf{T}_{\text{in}}=\mathbf{T}_{\text{out}}^c$, satisfies:
\begin{equation}
\sigma(\mathbf{T}_\text{in})=1-\sigma(\mathbf{T}_\text{out})\geq 1-2n\exp\left(-\frac{cn^\delta}{8}\right).
\end{equation}

We now introduce the primitive code book described in Section \ref{ssec:primitive}, denoted here by $\cC_P\subset S^{n-1}(R)$, with $N_P=|\cC_P|$. Notice that any random rotation of the space does not change the argument above. Indeed, if $U$ is a Haar-uniform random rotation in $SO(n)$, $\Pr\{UT\in \mathbf{T}_{\text{out}}\}\leq 2n\exp\left(-cn^\delta/8\right)$, by the same arguments as before. Hence, if $N(U):=|U\cC_P\in\mathbf{T}_\text{in}|$ is the number of words in a primitive code book rotated by $U$ which are inside the cube, by the linearity of the expectation, 
\begin{equation}
    \EE_U[N(U)]=N_P\cdot\sigma(\mathbf{T}_\text{in})\geq N_P\left(1-2ne^{-cn^\delta/8}\right).
\end{equation}
This implies that there must exist at least one rotation $U_0\in SO(n)$ such that $N(U_0)\geq N_P\left(1-2ne^{-cn^\delta/8}\right)$. 

Alternatively, one can show that, with overwhelming probability, a random rotation $U$ puts all but an exponentially small fraction of the rotated primitive code book inside the cube. Indeed, if we let $M(U):=N_P-N(U)$ be the primitive words outside the cube (and that, therefore, are unusable), we have that $\EE_U[M(U)]=N_P\cdot\sigma(\mathbf{T}_\text{out})\leq N_P2ne^{-cn^\delta/8}$. Invoking Markov's inequality,
\begin{equation}
\Pr\left\{M(U)\geq \epsilon N_P\right\}\leq\frac2\epsilon ne^{-cn^\delta/8}, 
\end{equation}
which holds for any $\epsilon>0$, and choosing $\epsilon:=2ne^{cn^\delta/16}$, one finally gets
\begin{equation}
\!\!\!\Pr\left\{M(U)\!\geq \!\epsilon N_P\!\right\}=\Pr\left\{N(U)\!\leq\! N_P(1-\epsilon)\right\}\leq e^{\frac{cn^\delta}{16}}\!.
\end{equation}

In Section \ref{ssec:error_analysis} we have shown that the proposed primitive code words can be reliably identified. In this section, we have observed that all but an exponentially small number of the primitive words are inside the unit cube that defines the input space of the Bernoulli channel and can therefore be used as code words for reliable deterministic identification:
\begin{theorem}\label{thm:Bernoulli}
Let $\lambda=2\exp\left[-2(\log n)^2\right]$, and fix $L\in \NN$ a large natural number independent of $n$. Then, there exists an $L$-layer reliable $(n,N,L\lambda,\lambda)$ code for DI over the Bernoulli channel that achieves a linearithmic rate 
\(
\dot{R}=\frac{1}{2}.
\)
\end{theorem}

\begin{proof}
Encode the $N:=N(U_0)$ messages into the $L$-layer points $\{\vec o_{s^L}:\vec o_{s^L}\in\mathbf{T}_\text{in}\}$, and to each message $i\in [N]$, associate the decoder $\cD_i=\cap_{\ell=1}^L\cD_i^{(\ell)}$, with $\cD_i^{(\ell)}$ given in Equation~\eqref{eq:layer-decoder}. The code is indeed reliable with the missed identification error probability behaving as calculated in the expressions \eqref{eq:primitive_layer_error1} and \eqref{eq:main_Pe1}, and false identification error probability behaving as in the expressions \eqref{eq:primitive_layer_error2} and \eqref{eq:main_Pe2}. It is only left to calculate the rate. 

The primitive code book size $N_P$ was bounded in Proposition~\ref{proposition:primitive_code_size}, together with $N(U_0)\geq N_P\left(1-2ne^{-cn^\delta/8}\right)$, a fact observed in this section, this produces
\begin{align}
\log N&\geq\log N_P+\log\left(1-2ne^{-cn^\delta/8}\right)\nonumber\\
&\geq \frac{n\log n}{2}\sum_{\ell=1}^L\frac{1-\delta}{2^\ell}-O(n\log\log n)\label{eq:size_step},
\end{align}
since the last term in the first line is absorbed by the $O(n\log\log n)$. The linearithmic rate is therefore given by
\begin{equation}\label{eq:linearithmic_rate}
\dot R(n):=\frac{\log N}{n\log n}\geq
\frac{1}{2}\sum_{\ell=1}^L\frac{1-\delta}{2^\ell}-O\left(\frac{\log\log n}{\log n}\right).
\end{equation}
Notice now that, choosing a sufficiently large $L$ and small $\delta>0$, one can make the sum $\sum_{\ell=1}^L{2^{-\ell}}$ arbitrarily close to $1$. Also, for large enough $n$ the last term is arbitrarily small. Therefore, given any $\epsilon>0$ one can choose a suitable $L$ and find a threshold $n_0$ such that, for all $n\ge n_0$, the linearithmic rate fulfils $\dot R(n)\geq\frac12-\epsilon$. Hence, $\dot R=\frac12$ is achievable, and the proof is completed.
\end{proof}

\begin{corollary}\label{cor:capacity}
The linearithmic deterministic identification capacity of the Bernoulli channel is $\dot C_\text{DI}(B)=\frac12$.
\end{corollary}
\begin{proof}
We observed in Theorem~\ref{thm:Bernoulli} that $\dot R=\frac12$ is achievable, so $\dot{C}_\text{DI}(B)\geq\frac12$. We know from \cite{CDBW:DI_classical}, c.f.~Equation~\eqref{eq:gap_intro}, that $\dot{C}_\text{DI}(B)\leq\frac12$. So, necessarily, $\dot{C}_\text{DI}(B)=\frac12$. 
\end{proof}

\subsection{Generalization to broader families of channels} \label{ssec:restricted_bern}

In the previous section we have closed the capacity gap for the Bernoulli channel, establishing that the linearithmic DI capacity is $\dot{C}_\text{DI}(B)=\frac12$. In this section, we generalize this result first to the restricted Bernoulli channel, and then to the family of channels $Q$ with continuous input and discrete output whose output is the union of a finite number of continuous 
curves and a residual set of (upper) Minkowski dimension at most $1$, and thus has $d_M=1$; as well as to the Gaussian and Poisson cases.

\begin{theorem}\label{thm:Restricted_Bernoulli}
Let $\lambda=2\exp\left[-2(\log n)^2\right]$, and fix $L\in \NN$. Then, there exists an $(n,N,L\lambda,\lambda)$ code for DI over the Bernoulli channel restricted to a continuous interval $[a,b]$ that achieves a linearithmic rate 
\(
\dot{R}=\frac{1}{2}.
\)
Hence, 
\(
\dot{C}_\text{DI}(B|_{[a,b]})=\frac12.
\)
\end{theorem}
\begin{proof}
Start by noticing that the available input space of the Bernoulli channel $B|_{[a,b]}$ restricted to a continuous interval $[a,b]$ is the $n$-dimensional cube of side $s:=b-a$. Let $Q$ be the centre of this cube. Then, the primitive code proposed on Section \ref{ssec:primitive} displaced to be concentric with the cube, $\vec o_{s_0}:=Q$, can be reused. Notice that displacing the primitive code does not change its size or the error analysis provided in Section~\ref{ssec:error_analysis}.

We need to be a bit more careful when calculating the amount of primitive words that are inside the (now smaller) cube of side $s$. The condition that the marginal distributions $T_i$ have to satisfy to be inside the cube is now $|T_i|\leq\frac{s}{2}$. Hence, in Equation~\eqref{eq:sub-gaussian_tail}, one needs to choose $t=\frac s2$. The following probability bound is thus produced, c.f.~\eqref{eq:out_cube_1Dcondition}:
\begin{equation}
\Pr\left\{|T_i|> \frac s2\right\}\leq 2\exp\left(-\frac{cs^2n^\delta}{8}\right).
\end{equation}

Then, following the exact same steps as in Section~\ref{ssec:opt_code}, one finds that the number $N$ of primitive words inside the cube of side $s$, which that can therefore be used for a reliable DI code over $B|_{[a,b]}$, satisfies:
\begin{equation}
N\geq N_P\left[1-2n\exp\left(-\frac{cs^2n^\delta}{8}\right)\right].
\end{equation}
Thus, $N\geq\frac12 n\log n\sum_{\ell=1}^L 2^{-\ell}-O(n\log\log n)$, since for any fixed $s>0$ the last term above is absorbed by the $O(n\log\log n)$ like we observed in Equation~\eqref{eq:size_step}. The linearithmic rate is therefore equal to Equation~\eqref{eq:linearithmic_rate}. Hence, repeating the same argument as in the proof of Theorem~\ref{thm:Bernoulli} we conclude that $\dot R=\frac12$ is an achievable rate. 

Finally, similarly to Cor.~\ref{cor:capacity}, $\dot R=\frac12$ being achievable implies that $\dot C_\text{{DI}}(B|_{[a,b]})\geq\frac12$, and as $\dot C_\text{{DI}}(B|_{[a,b]})\leq\frac12$, see \cite{CDBW:DI_classical}, it can be concluded that $\dot C_\text{{DI}}(B|_{[a,b]})=\frac12$.
\end{proof}

\begin{corollary}\label{cor:capacity_W}
The linearithmic DI capacity of any Bernoulli-reducible channel $W$ is bounded by $\dot C_\text{DI}(W)\geq\frac12$. Hence, for the particular family of channels $Q$ with continuous input, discrete output and $d_M=1$ described above, the Poisson channel $P$, and the Gaussian channel $G$ it holds 
\[
\dot C_\text{DI}(Q)=\dot C_\text{DI}(P)=\dot C_\text{DI}(G)=\frac12
\]

\end{corollary}
\begin{proof}
The first part follows immediately from Lemma~\ref{lemma:B_to_W}. Indeed, since the capacity of the restricted Bernoulli channel is $\dot C_\text{DI}(B)=\frac12$, we have that for any Bernoulli reducible channel $W$,
\begin{equation}
\dot{C}_{\text{DI}}(W) \geq \dot{C}_{\text{DI}}(B|_{[a,b]})=\frac12.
\end{equation}

Then, since the family of channels $Q$, the Poisson channel, and the Gaussian channel are indeed Bernoulli-reducible, and we know a matching converse result --see Equation~\eqref{eq:gap_intro} for the Gaussian and the Poisson cases, and Equation~\eqref{eq:gap2_intro} for the family of channels $Q$--, necessarily, $\dot C_\text{DI}(Q,P,G)=\frac12$.
\end{proof}


\section{Special example: the Poisson channel}\label{sec:poisson}
In this section we study in more depth the particular case of the \emph{Poisson channel} $P$. While its capacity has already been established by Corollary~\ref{cor:capacity_W}. The argument there, however, is purely existential. It only guarantees the existence of capacity-achieving DI codes via pre- and post-processing but does not describe them. Due to the theoretical and practical importance of the Poisson channel, we provide in this section an explicit optimal construction for DI over the Poisson channel under a peak power constraint. 

The Poisson channel is fundamental model in information theory, which captures communication scenarios in which the receiver observes discrete event counts generated by an underlying intensity parameter. Specifically, given a non-negative input $x\in\RR_{\geq 0}$, the output $Y=\{0,1,2,\dots\}$ is a non-negative integer with probability
\begin{equation}\label{eq:poisson}
\Pr(Y=y|x)=P_x(y)=e^{-x}\frac{x^y}{y!}.
\end{equation}
This channel arises naturally in physical settings where information is carried through random arrivals of particles. Consequently, it has become a central paradigm for optical communication systems \cite{D:optical_coms,DP:optical_coms,MS:optical_coms}, particularly in photon-limited regimes where photon counting is the primary mode of reception; for molecular communications\cite{Review:molecular_coms,tutorial:molecular_coms}, where information is conveyed through the release and sensing of molecules; and in other goal-oriented sensing and communication architectures that produce count outputs \cite{Poisson-bits-through-queues,S:Poisson-like-process}.

Specifically, in the DI setting, the Poisson channel has been studied under average or peak power constraints \cite{DI-poisson_mc,CDBW:DI_classical} and even in models with memory and feedback \cite{DI-poisson-isi,ZCBD:DI_isi_feedback}. In all these cases however, the construction is based in the sub-optimal typicality arguments, therefore producing the gap between lower and upper rate and capacity bounds [as evidenced in Equation~\eqref{eq:gap_intro}]. 

We will construct the code modifying our geometrical construction. For this construction to work, one needs to show that a concentration property equivalent to Proposition~\ref{proposition:projective_property} (which is, in turn, equivalent to that in \cite[Prop.~2]{CDBW:Gauss}) holds. In other words, one needs to ensure that the probability of the projection of the noise onto any fixed direction being larger than some threshold is small. 

\begin{proposition}\label{proposition:poisson}
Let $\vec{v}\in\RR^n$ with $\|\vec v\|>0$, and let $\vec{Z}=(Z_1,\dots,Z_n)$ denote the centred noise induced by the Poisson channel, i.e.\ $Z_i:=Y_i-x_i$ where $Y_i\sim\mathrm{Pois}(x_i)$. Then, for any input sequence satisfying the peak-power constraint $x^n\in[0,A]^n$ and any $t>0$,
\[
\Pr\left\{\|\Pi_{v}\vec Z\|>t\right\}\leq 2\exp \left(-\frac32t+\frac{9}{4}A\right).
\]
\end{proposition}
\begin{proof}
Similarly to what we did in the proof of Proposition~\ref{proposition:projective_property}, introduce an arbitrary non-null vector $\vec v$, with $\|\vec v\|>0$, normalized by $\vec u=\vec v/\|\vec v\|$. Then the orthogonal projection of the noise vector $\vec Z=(Z_1,\dots,Z_n)$ onto $\vec v$ is $\Pi_v\vec Z=\langle \vec Z,\vec u\rangle\vec u$, cf.~Equation~\eqref{eq:orth_projection}. 
Let $S_u:=\langle\vec Z,\vec u\rangle=\sum_{i=1}^nZ_iu_i$ be the normalized inner product, which is a sum of independent random variables $S_i=Z_iu_i$. In contrast to what we observed for the Bernoulli channel, however, these $S_i$ are unbounded (since, the noise for the Poisson channel can be arbitrarily large with positive probability). Therefore, one cannot apply Hoeffding's inequality. Instead, we use the Chernoff bound \cite{Chernoff_bound}, which given a $t>0$, and for any $\mu>0$, guarantees
\begin{equation}\label{eq:chernoff}
\Pr\{S_u\geq t\}=\Pr\{e^{\mu S_u}\geq e^{\mu t}\}\leq e^{-\mu t}\EE[e^{\mu S_u}].
\end{equation}
It remains to bound the moment generating function $\EE[e^{\mu S_u}]$. Start by noticing the following property of the Poisson distribution. For any $\beta\in\RR$, developing the expectation value first, and using the definition of the exponential function as the sum $e^\alpha=\sum_{k=1}^\infty\alpha^k/k!$ after, one gets
\begin{equation}
\EE[e^{\beta Y_i}]=e^{-x_i}\sum_{y=0}^\infty \frac{(e^{\beta}x_i)^y}{y!}=\exp[x_i(e^\beta-1)].
\end{equation}
Therefore, $\EE[e^{\beta Z_i}]=e^{-\beta_ix_i}\EE[e^{\beta Y_i}]=\exp[x_i(e^\beta-1-\beta)]$, and choosing $\beta:=\mu e_i$,
\begin{equation}\label{eq:poisson_step}
\!\!\!\EE[e^{\mu S_u}]\!=\!\prod_{i=1}^n\EE[e^{\mu u_iZ_i}]\!=\!\exp\!\left[\sum_{i=1}^nx_i(e^{\mu u_i}-1-\mu u_i)\!\right]\!.
\end{equation}
Use now the standard Bernstein-type bound $e^s-1-s \leq s^2$, this is a recurrently used tool in probability concentration inequalities \cite{Vershynin:high_dimensions_book} valid for $s\in[0,3/2]$. Then, since $x_i\leq A$ and $\sum_{i=1}^nu_i^2=1$, with $|u_i|\leq 1$, the sum in equation above becomes
\begin{equation}\label{eq:poisson_step2}
\sum_{i=1}^nx_i(e^{\mu u_i}-1-\mu u_i)\leq \sum_{i=1}^nx_i(\mu u_i)^2\leq\mu^2 A.
\end{equation}
Hence, putting together Equations \eqref{eq:chernoff}, \eqref{eq:poisson_step} and \eqref{eq:poisson_step2}, we obtain:
\begin{equation}
\Pr\{S_u\geq t\}\leq e^{-\mu t}\EE[e^{\mu S_u}] \leq \exp\left(-\mu t+\mu^2 A\right).
\end{equation}
Choosing $\mu=\frac32$, applying the same bound to $-S_u$, and by virtue of the union bound,
\begin{equation}
\Pr\left\{|S_u|\geq t\right\}\leq2 \exp\left(-\frac32 t+\frac94 A\right), 
\end{equation}
Hence, since we have that $\|\Pi_v\vec Z\|=\|S_u\vec u\|=|S_u|$, the fact $\Pr\left\{\|\Pi_{v}\vec Z\|>t\right\}=\Pr\left\{|S_u|\geq t\right\}$ completes the proof.
\end{proof}

Observe now the power of the hierarchical multi-layer construction defined in the previous section. As it is purely geometrical it can be reused to match the Poisson case.

\begin{theorem}\label{thm:Poisson}
Let $\lambda=\exp\left[\frac{3A}{2}\left(\frac32-\log n\right)\right]$, and fix $L\in \NN$. Then, there exists an $(n,N,L\lambda,\lambda)$ code for DI over the Poisson channel under the peak power constraint $A$ that achieves a linearithmic rate 
\(
\dot{R}=\frac{1}{2}.
\)
Hence, 
\(
\dot{C}_\text{DI}(P)=\frac12.
\)
\end{theorem}
\begin{proof}
Let us start just by multiplying the code parameters (the radii and the minimum projected distance $d$) by $A$ and noticing that the resulting code is just a scaled version of the code for the Bernoulli channel. Indeed, choosing $r_\ell=An^{(1-\delta)/2^\ell}$ and $d=3A\log n$, all the angles $\theta_\ell$ are the same as in the Bernoulli code proposed above, since they depend on the ratio $r_\ell/d$ and thus the scale factor $A$ cancels (cf.~Equation~\ref{eq:angles}). 

The size of the primitive code, calculated in Proposition~\ref{proposition:primitive_code_size}, is thus equal to the Bernoulli case, and so is the whole expurgation argument in Section~\ref{ssec:opt_code} where we removed the Bernoulli primitive words outside the input cube of side $1$, which is equivalent to expurgating a Poisson code (equal to the Bernoulli but scaled up by $A$) from the Poisson input cube of side $A$. We can conclude that the size $N_\text{Pois}$ of the $L$-layer code proposed in this paper satisfies (cf.~Equation~\eqref{eq:size_step}):
\begin{equation}
\log N_{\text{Pois}}\geq\frac{n\log n}{2}\sum_{\ell=1}^L\frac{1-\delta}{2^\ell}-O(n\log\log n).
\end{equation}
Since for any $\epsilon>0$ we can choose small $\delta>0$ and $L$ large enough such that there exists a threshold $n_0$ over which, for all $n>n_0$, we have $\dot R(n):=\log N_{\text{Pois}}/(n\log n)\geq\frac12-\epsilon$; we can conclude that $\dot R=\frac12$ is an achievable rate.

It only remains to perform the error analysis to ensure that the code above is reliable. Let us define the following scaled decoder at each layer
\begin{equation}\label{eq:layer-decoder-poisson}
\cD_{i}^{(\ell)}=\{\vec{y}\in\RR^n:\|\Pi_{s_{\ell-1}\to s_{\ell}}\vec{y}-\vec{o}_{s^\ell}\|\leq A\log n\}.
\end{equation}
Then, the layer-wise error bounds in Equations \eqref{eq:primitive_layer_error1} and \eqref{eq:primitive_layer_error2} become, for any two distinct messages $i\neq j\in[N_{\text{Pois}}]$ encoded respectively into the words $\vec o_{s^L}$ and $\vec o_{\tilde s^L}$,
\begin{align}
    \Pr(\cD_i^{(\ell)}|\vec o_{s^L})&=\Pr\left\{\|\Pi_{s_{\ell-1}\to s_{\ell}}\vec{Z}\|\leq A\log n\right\},\\
    \Pr(\cD_i^{(\ell)}|\vec o_{\tilde s^L})&\leq\Pr\left\{\|\Pi_{s_{\ell-1}\to s_{\ell}}\vec{Z}\|\geq A\log n\right\}.
\end{align}
Using Proposition~\ref{proposition:poisson} and defining $\lambda:=\exp\left[\frac{3A}{2}\left(\frac32-\log n\right)\right]$ one obtains $\Pr(\cD_i^{(\ell)}|\vec o_{s^L})\leq1-\lambda$ and $\Pr(\cD_i^{(\ell)}|\vec o_{\tilde s^L})\leq \lambda$. Thus, by the union bound and the same arguments as in Equations \eqref{eq:main_Pe1} and \eqref{eq:main_Pe2} we obtain that, for any two distinct messages $i\neq j\in[N_{\text{Pois}}]$,
\begin{equation}
    P_{e,1}(i)\leq L\lambda:=\lambda_1\,\,\text{and}\,\,P_{e,2}(i|j)\leq \lambda:=\lambda_2.
\end{equation}
Both $\lambda_1$ and $\lambda_2$ vanish as $n\to\infty$ so the code is reliable for DI over the Poisson channel. 

We complete the proof by noticing that, since $\dot R=\frac12$ is an achievable rate, then $\dot{C}_\text{DI}(P)\geq\frac12$; and since we know from \cite{CDBW:DI_classical}, cf.~Equation~\eqref{eq:gap_intro}, that $\dot{C}_\text{DI}(P)\leq\frac12$, we necessarily have $\dot{C}_\text{DI}(P)=\frac12$.
\end{proof}



\section{Rate-reliability tradeoff}\label{sec:RR}
Until now, we have focused on the linearithmic DI capacity by analysing the asymptotic performance of the proposed code at first order. The particular parameter choices, such as the layer radii $r_\ell=n^{(1-\delta)/2^\ell}$, the projective distance $d=3\log n$, and the acceptance region of the decoder in \eqref{eq:layer-decoder}, were designed to ensure that the linearithmic rate approaches $\frac12$ while both error probabilities vanish as $n\to\infty$.

Beyond the optimal rate performance characterised by capacity, it is important in many applications to understand the speed at which the error probabilities decay. In classical Shannon-style message transmission it is well known that, when communicating at rates $0\leq R<C_T$ below capacity, the error probability decays exponentially fast with block length, i.e.\ $P_e\approx 2^{-nE_T(R)}$ \cite{Shannon:TheoryCommunication,Gallager:ReliabilityBook}. In contrast, for DI, it was observed that linearithmic scaling is only possible when both error probabilities vanish sub-exponentially fast \cite{CDBW:Reliability-TCOM,RRGauss-arXiv}. Indeed, the error probabilities decay exponentially fast, i.e.\ $\lambda_1\sim 2^{-nE_1}$ and $\lambda_2\sim 2^{-nE_2}$ with constant exponents (actually, even if only one of the two does \cite{DI-steins}), then the identification rate reverts to the linear scale.

Initial rate-reliability studies of DI \cite{CDBW:Reliability-TCOM,RRGauss-arXiv,DI-steins} provided upper and lower bounds on the trade-off between rate and error decay. However, the achievability bounds were based on typicality-based constructions, and consequently there appeared a gap with respect to the converse bounds, already at leading order.

An analysis on the tradeoff between error and rates of the capacity-achieving Gaussian DI code (the geometrical construction that inspired the present work) was presented in \cite{CDBW:Gauss}. There, the lower bound on the rate-reliability tradeoff function was shown to match the upper bound at first order, thereby closing the previously observed leading-order gap. 

In this section we will study the tradeoff between rate and errors in the proposed construction for DI over channels $Q$ with $d_M=1$ continuous input and discrete output, showing that we can match at first order the upper bound (calculated in \cite{CDBW:Reliability-TCOM} and reproduced in Theorem~\ref{thm:RR_converse} below) in the regime of small error exponents. 
\medskip

Let us start with some necessary preliminaries. Given an $(n,N^{(n)},\lambda_1^{(n)},\lambda_2^{(n)})$ DI code, we define the error exponents $E_1(n)$, $E_2(n)$, and the rate $R(n)$ in the \emph{linear scale} as follows
\begin{equation}
\lambda_1^{(n)}\!=2^{-nE_1(n)}\!,\,\,\lambda_2^{(n)}\!=2^{-nE_2(n)}\!,\,\,R(n)\!=\!\frac{\log N^{(n)}}{n}.
\end{equation}
Notice that, as we need the errors to vanish in the asymptotic limit $\lim_{n\to\infty}\lambda_1^{(n)},\lambda_2^{(n)}=0$, it is necessary that the error exponents satisfy $E_1(n),E_2(n)=\omega(1/n)$.

\begin{theorem}[cf. Cor. 8 \cite{CDBW:Reliability-TCOM}]\label{thm:RR_converse}
Given a channel $Q$ with $d_M\!=\!1$ and for all $\eta>0$ there exists an $E_0>0$ such that for all $E_1(n),E_2(n)\geq E(n)$ with $E(n)\leq E_0$ and sufficiently large $n$ it holds
\[
R(n)\leq (1+\eta) \log\frac{2}{\sqrt{1-e^{-E(n)/2}}}.
\]
Moreover, in the regime of small $E(n)$ the expression above can be expanded as follows,
\[
R(n)\leq\frac{1+\eta}{2}\log\frac{8}{E(n)}+O[E(n)].
\]
\end{theorem}
The theorem shows that if the error probabilities decay exponentially fast, i.e.\ $E(n)=\Theta(1)$, then the linear-scale rate $R(n)$ remains bounded, and linearithmic scaling is impossible. On the other hand, if $E(n)$ decreases with $n$, corresponding to sub-exponential error decay, then the upper bound on $R(n)$ grows with $n$.

The theorem above tells us that, if the errors vanish exponentially fast, i.e. $E(n)=\text{cnst.}$, the linear scale rate $R(n)$ remains bounded, and thus it is not possible to obtain a linearithmic scale on the rates (as we claimed in the beginning of this section). On the other hand, if $E(n)$ decreases with $n$, in other words, the errors decay sub-exponentially fast, the linear rate grows with $n$ and therefore diverges on the limit $n\to\infty$, indicating a faster scaling of the rate. 

In particular, in the regime of smallest possible error exponents $E(n)=\omega(1/n)$ (the errors vanish very slowly) one observes that $R(n)\leq\frac{1+\eta}{2}\log n+O(1)$. Which implies that the linearithmic rate satisfies $\dot R(n)\leq\frac{1+\eta}{2}+O(\frac{1}{\log n})$. Hence, since the last term vanishes for increasing $n$ and $\eta$ can be chosen arbitrarily small, we recover the $\frac12$ leading factor on the linearithmic rate that defines the capacity upper bound. 

\begin{theorem}\label{thm:RR_Bern}
Let $\delta>0$ and $L>0$ be fixed constants, and $\eta(L)>0$ a vanishing function on $L$ (independent of $n$). Then, given any Bernoulli-reducible channel channel $W$ continuous input and discrete output one can construct an $(n,N,L2^{-nE(n)},2^{-nE(n)})$ DI code for all $E(n)<1/\delta\log n$ with linear rate bounded by
\[
R(n)\geq\frac{1-\eta(L)}{2}\log\left(\frac{1}{E(n)}-\delta\log n\right)-O(1).
\]
\end{theorem}

\begin{proof}
We follow the main ideas in the proof of the rate-reliability lower bound for the Gaussian case \cite[Thm.~15]{CDBW:Gauss}. Start by observing that the error performance in Section~\ref{ssec:error_analysis} is dictated by the acceptance range $t$ of the decoders at each layer in Equation~\eqref{eq:layer-decoder}, by virtue of Proposition~\ref{proposition:projective_property}. There, we chose $t=\log n$ and observed that the errors could be bounded in terms of the object $\lambda=2\exp(-2t^2)$.

Instead, if one now chooses the decoding acceptance range to be $t=\sqrt{nE(n)}$, i.e.,
\begin{equation}
\cD_{i}^{(\ell)}=\left\{\vec{y}\in\RR^n:\|\Pi_{s_{\ell-1}\to s_{\ell}}\vec{y}-\vec{o}_{s^\ell}\|\leq\sqrt{nE(n)}\right\},
\end{equation}
the error probabilities at each layer [cf. Equations \eqref{eq:primitive_layer_error1} an \eqref{eq:primitive_layer_error2}] are bounded, using Proposition \ref{proposition:projective_property}, by the following object
\begin{equation}
2e^{-2nE(n)}\leq2^{-nE(n)}:=\lambda.
\end{equation}
In the last equality we have redefined $\lambda$. Then, the primitive errors across all layers are $\lambda_1:=L\lambda=2^{-nE(n)+\log L}$ and $\lambda_2:=\lambda=2^{-nE(n)}$. In other words $E_2(n)=E(n)$ and $E_1(n)=E(n)-\frac1n\log L$.

Notice also that for the second type of error to work [concretely, the step $(\gamma)$ in Equation \eqref{eq:primitive_layer_error2}] we need to make sure that Proposition \ref{proposition:concentration}, specifically Equation~\eqref{eq:concentration_property}, works as intended. In other words, we need,
\begin{equation}\label{eq:RR_concentration_property}
\|\Pi_{s_{\ell-1}\to s_\ell}\vec o_{\tilde s^L}-\vec o_{s^\ell}\|\geq 2t.
\end{equation}

Looking at Proposition \ref{proposition:concentration} one observes that the condition above is ensured by choosing $d:=3t$ and imposing $\Delta\leq t$. Let us choose $\sqrt{L}r_{\ell+1}\sin\theta_\ell=t$, then $\Delta\leq\sqrt{L}r_{\ell+1}\sin\theta_\ell\leq t$. Therefore, since in the regime of large $n$ one can approximate $\sin\theta_\ell\approx\sqrt{\frac{2d}{r_\ell}}$, we can equate the $\sin\theta_\ell$ in the last expressions to find:
\begin{equation}
\frac{6t}{r_\ell}=\frac{t}{r_{\ell+1}\sqrt{L}}\implies r_{\ell+1}=\sqrt{\frac{tr_\ell}{6L}}.
\end{equation}

Choosing now the radius at the first layer $r_1:=n^{\frac{1-\delta}{2}}$, we can extract the radius at each layer:
\begin{equation}\label{eq:RR_radius_layer}
    r_\ell=\frac{t}{6L}\left(\frac{n^{1-\delta}}{(t/6L)^2}\right)^{\frac{1}{2^\ell}}.
\end{equation}
Then, one can bound the amount of elements in the primitive code book using Proposition~\ref{proposition:primitive_code_size}. Indeed, at layer $\ell$ one finds
\begin{align}
\log N_\ell&\geq\frac{n-\ell+1}{2}\log\left[\frac{2t}{6L}\frac{1}{3t}\left(\frac{n^{1-\delta}}{(t/6L)^2}\right)^\frac{1}{2^\ell}\right]\nonumber\\
&=\frac{n-\ell+1}{2}\log\left[\frac{(6L)^{\frac{1}{2^{\ell-1}}}}{9L}\left(\frac{n^{1-\delta}}{nE(n)}\right)^\frac{1}{2^\ell}\right]\nonumber\\
&=\frac{n}{2}\log\left(\frac{1}{n^\delta E(n)}\right)^\frac{1}{2^\ell}-O(n)\nonumber\\
&=\frac{n}{2}\frac{1}{2^\ell}\log\left(\frac{1}{E(n)}-\delta\log n\right)-O(n)\label{RR_layer_size}.
\end{align}
Notice that for the bound above to be meaningful we need $E(n)<1/\delta\log n$.
Now, since the total amount of primitive words is given by $N_P=\prod_{\ell=1}^LN_\ell$, one obtains:
\begin{equation}\label{RR_primitive_size}
\log N_P\geq\frac{n}{2}\sum_{\ell=1}^L\frac{1}{2^\ell}\log\left(\frac{1}{E(n)}-\delta\log n\right)-O(n).
\end{equation}

Using the arguments in Section~\ref{ssec:opt_code}, we know there must exist at least one rotation $U_0\in SO(n)$ such that the number $N(U_0)$ of elements in a good reliable code for DI over the Bernoulli channel is bounded by $N(U_0)\geq N_P(1-2n\exp(-cn^\delta/8))$ for some constant $c>0$. Similarly, for the Bernoulli restricted to a continuous interval $[a,b]$, with $s:=b-a$, the number of elements can be bounded by $N_{[a,b]}(U_0)\geq N_P(1-2n\exp(-cs^2n^\delta/8))$ (see Section \ref{ssec:restricted_bern}). Since the negative term of the size correction is of order $\exp [-\text{poly}(n)]$, it is absorbed by the $O(n)$ term in Equation~\eqref{RR_primitive_size}. Therefore, using Lemma~\ref{lemma:B_to_W}, one can extract that the number $N$ of elements in a reliable $(n,N,L\lambda,\lambda)$ DI code, with $\lambda=2^{-nE(n)}$ is bounded by
\begin{equation}\label{RR_size}
\log N\geq\frac{n}{2}\sum_{\ell=1}^L\frac{1}{2^\ell}\log\left(\frac{1}{E(n)}-\delta\log n\right)-O(n).
\end{equation}

Finally, the linear rate $R(n):=\frac1n\log N$ is given by
\begin{equation}\label{RR_linear_rate}
R(n)\geq\frac{1}{2}\sum_{\ell=1}^L\frac{1}{2^\ell}\log\left(\frac{1}{E(n)}-\delta\log n\right)-O(1).
\end{equation}
The proof is completed by defining  $\eta(L):=1-\sum_{\ell=1}^L\frac{1}{2^\ell}>0$ and noticing that, since $\sum_{\ell=1}^L\frac{1}{2^\ell}$ can be arbitrarily close to 1 if we choose $L$ large enough, then $\eta(L)$ is a constant that, through the choice of $L$, can be selected arbitrarily small.  
\end{proof}

Notice that, in the regime $E(n)=o(1/\log n)$ the lower bound on the rate-reliability function in Theorem~\ref{thm:RR_Bern} matches at first order the upper bound in Theorem~\ref{thm:RR_converse}. Indeed, by choosing small values for the constants we find:
\begin{equation}
\begin{split}
R(n)&\geq\frac{1-\eta(L)}{2}\log\frac{1}{E(n)}-O[E(n)\log n],\\
R(n)&\leq\frac{1-\eta}{2}\log\frac{1}{E(n)}+O(1),
\end{split}
\end{equation}
In other words, since the constants $\eta(L)$ and $\eta$ can be chosen arbitrarily small, for any $\epsilon>0$ and sufficiently large $n$ one finds, for any Bernoulli reducible channel with $d_M=1$,
\begin{equation}
   \frac{1}{2}\log\frac{1}{E(n)}-\epsilon \leq R(n)\leq \frac{1}{2}\log\frac{1}{E(n)}+\epsilon.
\end{equation}

We have already mentioned that the rate-reliability construction above becomes trivial in the regime of exponentially-fast vanishing errors. Indeed, in Theorem~\ref{thm:RR_Bern}, we imposed the condition $E(n)<1/\delta \log n$, otherwise the linear rate becomes negative. However, we know from \cite{CDBW:Reliability-TCOM} that there exist lower bounds to the linear rate-reliability function for DI over the Bernoulli channel which, while they exhibit a gap with respect to the upper bound, they are positive in the whole range of errors $E(n)\in\left[O(1),\omega(\frac1n)\right)$. Hence, accepting also the constant error exponent (i.e.~exponentially fast vanishing error) regime. 

In other words, while the construction proposed above performs better in the regime $E(n)\in\left(\frac{1}{\delta\log n},\omega(\frac1n)\right)$, closing the gap observed in \cite{CDBW:Reliability-TCOM} at leading order, it gives us no meaningful information on the regimes $E(n)\geq1/\delta\log n$. Why does this happen? and why it does this effect not appear in the rate-reliability analysis performed for the similar Gaussian construction in \cite{CDBW:Gauss}?

Well, the restricting condition stems from the fact that the radius at each layer is $r_\ell=n^{(1-\delta)/2^\ell}$, while in the similar construction for the Gaussian channel we could choose $\tilde r_\ell=(cn)^{1/2^\ell}$ with $0<c<E$ a positive constant smaller than the Gaussian channel power constraint $E$. Notice that the radii $r_\ell$ used here is smaller by a polynomial factor $O(n^\delta)$, which is carried throughout the calculation generating the negative term $\delta\log n$ that makes the rate trivial for error exponents that are too large. As a matter of fact, one can calculate that the \emph{universal code} for DI over Gaussian channels proposed in \cite{CDBW:Gauss}, which uses the same radii $r_\ell=n^{(1-\delta)/2^\ell}$ would produce the same issue. 

Furthermore, notice that, for the Bernoulli channel we can not use the larger radii $\tilde r_\ell=(cn)^{1/2^\ell}$. That is because the resulting primitive code would suddenly fall mostly outside the cube and would therefore be unusable. Indeed, due to the application of the union bound in Equation~\eqref{eq:out_cube_Multi_D_condition}, the probability that a random point on the surface of the sphere of radius $\tilde R=\sqrt{\sum_{\ell=1}^L\tilde{r}_\ell^2}$ is also inside the input space cube diverges:
\begin{equation}
\!\!\Pr\left\{T\in \mathbf{T}_{\text{out}}\right\}\leq\sum_{i=1}^n\Pr\left\{|T_i|\!>\! \frac12\right\}\leq 2ne^{-k}\stackrel{n\to\infty}{\longrightarrow}\infty,
\end{equation}
where $k>0$ is some constant. So the polynomial radius reduction used in the construction here presented is necessary for the primitive code to be usable in Bernoulli channels and their reductions, but it causes a reduced applicability range on the rate-reliability functions.

In any case, one needs to notice that we are mostly interested in the regime of small error exponents, that is the valid range in Theorem~\ref{thm:RR_Bern}, because it is the only regime where we find a linearithmic scaling of the rates (the critical feature of DI that asymptotically improves upon Shannon-like settings). Indeed, as already mentioned previously, and as shown by the converse bound in Theorem~\ref{thm:RR_converse}, in the regimes where the error vanishes faster, the linearithmic scaling is lost.


\section{Conclusions and discussion}\label{sec:conclusions}
Building on top of the multi-layer geometric ideas introduced for the AWGN channel in \cite{CDBW:Gauss}, we have developed in this work an adapted geometrical construction for DI applicable to a large variety of channels and we have used it to resolve the long-standing rate and capacity gap in the linearithmic regime. In particular, we constructed an explicit DI code for the Bernoulli channel and proved that it achieves the known converse bound in \cite{CDBW:DI_classical}, thereby establishing the exact capacity
\begin{equation}
\dot C_{\text{DI}}(B)=\frac12.
\end{equation}
We then extended the construction to restricted Bernoulli channels $B|_{[a,b]}$ and, via a Bernoulli reduction argument, to any memoryless channel $W$ whose channel image contains a non-trivial continuous curve, thus establishing for this whole family that $\dot C_{\text{DI}}(W)\geq\frac12$. In particular, for the important family of channels $Q$ such that their image is the union of a finite number of continuous curves and a residual set of (upper) Minkowski dimension at most 1, for which the converse bound $\dot C_{\text{DI}}(Q)\leq \frac12$ was previously known \cite{CDBW:DI_classical}, our achievability result yields the exact capacity formula
\begin{equation}
    \dot C_{\text{DI}}(Q)=\frac12,
\end{equation}
thereby closing the persistent capacity gap identified in \cite{CDBW:DI_classical,CDBW:Reliability-TCOM} for channels with $d_M=1$. We also colsed the gap for the Poisson channel $\dot C_{\text{DI}}(P)=\frac12$ and further provided an explicit capacity-achieving construction under a peak-power constraint. The capacity results in \cite{CDBW:Gauss} are also reproduced by the generalization here.

Beyond capacity, we have also studied here the rate-reliability tradeoff associated with the proposed construction. We derived improved achievability bounds on the DI reliability function and showed that, in the regime of small error exponents, the leading-order behaviour matches the known converse bound, thereby closing gap observed at first order in the rate-reliability analysis \cite{CDBW:Reliability-TCOM,RRGauss-arXiv}.

The proposed optimal construction is not based on typicality arguments, but rather on high-dimensional geometric concentration effects. 
This confirms the intuition provided by previous approaches that identified typicality arguments as a probable source of the gap in capacity and rates.
Also, while the concentration methods initially proposed in \cite{CDBW:Gauss} could a priory be associated with isotropic noise models such as AWGN, our results show that the same approach can be adapted to channels where the noise is strongly input-dependent, such as the Bernoulli and Poisson channels. This suggests that geometric constructions may provide a unifying tool for DI over an even wider range of channel models.

Bernoulli-reducible channels do not exhaust all scenarios in which superlinear DI is possible. In particular, channels whose channel image has Minkowski dimension $d\neq 1$ (for instance, channels whose image contains higher-dimensional manifolds such as surfaces or volumes, or channels with fractal or otherwise disconnected image sets) may exhibit superlinear scaling behaviour governed by the $d_M$, as studied in \cite{CDBW:DI_classical}. Such channels lie outside the scope of the present paper: while our multi-layer geometric construction yields an achievable linearithmic rate of $\frac12$ whenever a Bernoulli reduction is possible, it is in general not optimal when the channel image has dimension $d>1$, since we know that the rates $d_M/4$ are achievable and the general converse bound scales as $d_M/2$. It remains an interesting open direction to determine whether multi-layer geometric constructions can be generalised to such higher-dimensional or irregular channel images, potentially leading to capacity-achieving schemes whose performance is characterised exactly in terms of Minkowski dimension.

Finally, notice the surprising universality of our capacity characterisation. In contrast to Shannon transmission, where capacity depends explicitly on metric and information-theoretic channel parameters (e.g., the signal-to-noise ratio in AWGN channels, or power constraints in Poisson channels), the DI capacity is largely insensitive to such details. Rather, it is governed by qualitative geometric features of the channel image in the space of output distributions. In particular, for the broad class of channels studied in this paper, the capacity is fixed at the universal value $\dot C_{\text{DI}}=\frac12$, independently of all other specific channel characteristics. Indeed, the channel parameters only influence lower-order terms, such as the precise error decay and finite block length behaviour.


\bibliographystyle{ieeetr}
\bibliography{ID}

\end{document}